\documentclass[runningheads]{llncs}
\usepackage[T1]{fontenc}

\usepackage[draft, todonotes={textsize=scriptsize}]{changes}
\setuptodonotes{color=red, backgroundcolor=red,
  bordercolor=red, tickmarkheight=10pt}
\definechangesauthor[name=Robert Beinert, color=cyan]{RB}
\definechangesauthor[name=Saghar Rezaei, color=green]{SR}

\usepackage{amsmath}
\usepackage{amsfonts}
\usepackage{mathtools}
\usepackage{nicefrac}
\usepackage{csquotes}

\usepackage{subcaption}

\usepackage{booktabs}
\usepackage[shortlabels]{enumitem}

\usepackage[hidelinks]{hyperref}
\usepackage{cleveref}
\usepackage{cite}

\usepackage{graphicx}
\usepackage{color}

\newcommand{\BR}{\bbbr}

\newcommand{\BN}{\bbbn}

\newcommand{\BC}{\bbbc}

\newcommand{\family}[1]{\mathfrak{#1}}
\newcommand{\fT}{\family{T}}
\newcommand{\fC}{\family{C}}
\newcommand{\fD}{\family{D}}
\newcommand{\fK}{\family{K}}

\newcommand{\e}{\mathrm{e}}
\newcommand{\I}{\mathrm{i}}

\newcommand{\diff}{\mathop{}\!\mathrm{d}}

\newcommand{\ip}[2]{\langle#1,#2\rangle}

\newcommand{\abs}[1]{\lvert#1\rvert}

\newcommand{\absB}[1]{\Bigl\lvert#1\Bigr\rvert}

\newcommand{\norm}[1]{\lVert#1\rVert}

\DeclareMathOperator{\fourier}{\mathcal{F}}
\DeclareMathOperator{\borel}{\mathcal{B}}
\DeclareMathOperator{\measure}{\mathcal{M}}

\newcommand{\CO}{C_0}
\newcommand{\Cb}{C_{\mathrm{b}}}

\spnewtheorem{algorithm}{Algorithm}{\bfseries}{\normalshape}

\begin{document}
\title{Prony-Based Super-Resolution Phase Retrieval of Sparse, Multivariate Signals%
  \thanks{Supported by
    the Federal Ministry of Education and Research (BMBF, Germany)
    [grant number 13N15754].
  }%
}
\titlerunning{Prony-Based Super-Resolution Phase Retrieval}
%
\author{Robert Beinert \and Saghar Rezaei}
\authorrunning{R. Beinert, S. Rezaei}
%

\institute{Technische Universität Berlin, Institute of Mathematics,\\
  Stra{\ss}e des 17. Juni 136, 10623 Berlin, Germany\\
  \email{robert.beinert@tu-berlin.de}, \email{s.rezaei@campus.tu-berlin.de}\\
  \url{https://tu.berlin/imageanalysis/}}

\maketitle              
\begin{abstract}
  Phase retrieval consists in the recovery of an unknown signal
  from phaseless measurements of its usually complex-valued Fourier transform.
  Without further assumptions,
  this problem is notorious to be severe ill posed
  such that the recovery of the true signal is nearly impossible.
  In certain applications
  like crystallography,
  speckle imaging in astronomy,
  or blind channel estimation in communications,
  the unknown signal has a specific, sparse structure.
  In this paper,
  we exploit these sparse structure
  to recover the unknown signal uniquely
  up to inevitable ambiguities as
  global phase shifts,
  transitions,
  and conjugated reflections.
  Although using a constructive proof
  essentially based on Prony's method,
  our focus lies on the derivation of a recovery guarantee
  for multivariate signals
  using an adaptive sampling scheme.
  Instead of sampling the entire multivariate Fourier intensity,
  we only employ Fourier samples along certain adaptively chosen lines.
  For bivariate signals,
  an analogous result can be established for samples in generic directions.
  The number of samples here scales quadratically
  to the sparsity level of the unknown signal.
  
\end{abstract}
\section{Introduction}
\label{sec:introduction}

Phase retrieval is one of the major challenges
in many imaging tasks in physics and engineering.  
For instance,
phase retrieval is an essential component of the imaging techniques:
ptychography \cite{KCPU16}, crystallography \cite{Millane:90}, speck\-le imaging \cite{Knox:76}, diffraction tomography \cite{BQ22}.
Although there are many different problem formulations summarized as \enquote{phase retrieval},
the main task consists in the recovery of an image or signal
form the magnitudes of usually complex-valued measurements.
If we think at the application in optics,
the measurements often correspond to the magnitudes of the Fourier or Fresnel transform.
Although the Fourier and Fresnel transform are invertible,
the loss of the phase turns the imaging task into an severe ill-posed inverse problem.
The central challenge is here the strong ambiguousness,
which has been studied for the continuous as well as for the discrete setting,
see for instance \cite{beinert2015ambiguities,BBE17,GroKopRat20}.
To overcome the ambiguousness,
different a priori informations
like support constraints or non-negativity
as well as additional measurements
have been studied \cite{KliKam14,KliSacTik95,BP18,Bei17,Bei17a}.

In this paper,
we are interested in the recovery of multivariate, sparse signals,
which are modeled as complex measure
\begin{equation*}
  \mu \coloneqq \sum_{n=1}^N c_n \, \nu(\cdot - T_n),
\end{equation*}
where $\nu$ is a known structure.
Choosing $\nu$ as Dirac measure,
we obtain point measures or spike functions.
But other choices like Gaussians are reasonable too
and lead to more regular functions.
We here consider the Fourier phase retrieval problem,
meaning that we want to recover $c_n$ and $T_n$
from samples of the Fourier intensity---%
the magnitude of the Fourier transform---%
of $\mu$.
Our focus lies on the derivation of uniqueness guarantees,
i.e.\ on assumptions allowing the determination of $c_n$ and $T_n$
up to unavoidable ambiguities.
Since the transitions $T_n$ may lie on a continuous domain,
our problem can be interpreted as super-resolution phase retrieval \cite{baechler2019super}.
This kind of problem arises in applications like
crystallography \cite{Millane:90},
speckle imaging in astronomy \cite{Knox:76}
as well as in
blind channel estimation in communication \cite{barbotin2012fast}.

\paragraph{Methodology and Relation to the Literature}

To show that
phase retrieval of sparse, multivariate signals is possible in principle,
we rely on the Prony-based techniques in \cite{beinert2017sparse},
where super-resolution phase retrieval is studied for univariate signals.
For the generalization to multivariate signals,
we employ an adaptive sampling strategy,
where the Fourier intensity of the unknown signal is
measured along adaptively chosen or along generic lines.
This sampling method traces back to \cite{PW13},
where the recovery of sparse signals
from their (complex-valued) Fourier transform is studied.
Our main idea is to use the specific sampling setup
to reduce the multivariate phase retrieval problem
into a series of univariate problems,
to solve these univariate instances as in \cite{beinert2017sparse},
and to combine the extracted informations
to solve the multivariate instance.
The considered super-resolution phase retrieval problem
has also been studied in \cite{baechler2019super},
where a greedy-like algorithm is proposed to solve the problem numerically.
In difference to \cite{baechler2019super},
where the entire Fourier domain is sampled,
we show that
the unknown $D$-variate signal is completely determined
by $\mathcal O(N^2)$ measurements on $2D-1$ lines,
where $N$ denotes the sparsity level.

\paragraph{Contribution}

The contribution of this paper is the derivation
of a recovery guarantee for super-resolution phase retrieval
of multivariate signals.
The main theorem shows that,
under mild assumptions like collision-freeness,
the unknown signal is completely determined by
equispaced samples on a few lines.
All in all,
we here require $\mathcal O(D N^2)$ samples of the Fourier intensity
to recover a $D$-variate signal composed of $N$ components.
This shows that,
at least theoretically,
it is enough to use a space sampling setup
to recover a sparse signal.
Although the proofs are constructive,
the focus lies on the theoretical uniqueness guarantee
since the applied Prony method is known to be unstable
for noisy measurements.

\paragraph{Outline}

Before considering the super-resolution phase retrieval problem,
we briefly introduce the needed concepts like
the Fourier transform of measures and
Prony's method in §~\ref{sec:preleminaries}.
In §~\ref{sec:phase-retr-probl},
we define the considered problem in more details
and discuss the unavoidable ambiguities.
Section~\ref{sec:phase-retrieval-line} is devoted to
the univariate sparse phase retrieval problem,
which we generalize to the multivariate setting in §~\ref{sec:phase-retrieval-real-space}.
The numerical simulation in §~\ref{sec:simulations} show that
the constructive proofs can be implemented
to recover the unknown signal at least in the noise-free setting.
A final discussion is given in §~\ref{sec:conclusion}.

\section{Preleminaries}
\label{sec:preleminaries}

\subsection{Fourier Transform of Measures}
\label{sec:four-transf-meas}

Subsequently,
the unknown signals are characterized as complex measures.
For this,
let $\borel(\BR^D)$ denote the Borel $\sigma$-algebra
of the Euclidean space $\BR^D$,
and
$\measure(\BR^D)$ the space of all regular, finite, complex measures.
Every considered signal is then interpreted as complex-valued mapping
$\mu \colon \borel(\BR^D) \to \BC$
with $\mu \in \measure (\BR^D)$.
Recall that
$\measure(\BR^D)$ is the dual space of $\CO(\BR^D)$,
which consists of all continuous functions $\phi\colon \BR^D \to \BC$
where $\phi(x)$ vanishes for $\norm{x} \to \infty$.
Every measure $\mu \in \measure(\BR^D)$ hence defines
a continuous, linear mapping via
\begin{equation*}
  \phi
  \mapsto
  \ip{\mu}{\phi}
  \coloneqq \int_{\BR^D} \phi(x) \diff \mu(x).
\end{equation*}
The convolution of two measures
$\mu, \nu \in \measure(\BR^D)$
is indirectly defined as
\begin{equation*}
  \ip{\mu * \nu}{\phi}
  \coloneqq
  \iint_{\BR^D \times \BR^D} \phi(x+y) \diff \nu(y) \diff \mu(x),
  \quad
  \phi \in \CO(\BR^D).
\end{equation*}
The Fourier transform on $\measure(\BR^D)$
is given by
$\fourier \colon \measure(\BR^D) \to \Cb(\BR^D)$
with
\begin{equation*}
  \fourier [\mu](\omega)
  \coloneqq
  \hat \mu(\omega)
  \coloneqq
  \int_{\BR^D} \e^{-\I \ip{\omega}{x}} \diff \mu(x),
  \quad
  \omega \in \BR^D,
\end{equation*}
where $\Cb(\BR^D)$ consists of all bounded, continuous functions on $\BR^D$.
Notice that
$\fourier$ is continuous and injective.
Further,
the \emph{Fourier convolution theorem} states
\begin{equation*}
  \fourier[\mu * \nu]
  =
  \hat \mu \hat \nu.
\end{equation*}

\subsection{Prony's Method}
\label{sec:pronys-method}

The Fourier intensities of a sparse signal on $\BR$ are essentially given
by a non-negative exponential sum $E \colon \BR \to \BR$
of the form
\begin{equation}
  \label{eq:exp-sum}
  E(\omega)
  \coloneqq
  \sum_{\ell=-L}^{L} \gamma_\ell \, \e^{- \I \omega \tau_\ell}
  =
  \gamma_{0}
  + \sum_{\ell=1}^{L}
  \bigl( \gamma_{\ell} \, \e^{-\I\omega\tau_{\ell}}
  + \bar \gamma_{\ell} \, \e^{\I \omega \tau_{\ell}} \bigr)
\end{equation}
with $\gamma_{\ell} = \bar \gamma_{-\ell} \in \BC\setminus\{0\}$
and $\tau_{\ell} = - \tau_{-\ell} \in \BR$
for $\ell = 0, \dots, L$.
Note that $\gamma_{0} \in \BR\setminus\{0\}$ and $\tau_{0}=0$.
For pairwise distinct $\tau_{\ell}$,
the parameters $\gamma_{\ell}$ and $\tau_{\ell}$ may be recovered
by Prony's method \cite{de1795essai,Hil87,potts2010parameter}.
In a nutshell,
for $h > 0$ with $h \tau_\ell < \pi$,
$\ell = 1, \dots, L$,
and for $M \ge 4L+1$,
we define the Prony polynomial
\begin{equation}
  \label{eq:prony-poly}
  \Lambda(z)
  \coloneqq
  \prod_{\ell=-L}^L
  \bigl( z - \e^{-\I h \tau_{\ell}} \bigr)
  =
  \sum_{k=0}^{2L+1} \lambda_k \, z^k,
  \quad
  z \in \BC,
\end{equation}
and observe
\begin{equation*}
  \sum_{k=0}^{2L+1} \lambda_k \, E(h(k+m))
  =
  \sum_{\ell=-L}^{L}
  \gamma_\ell \, \e^{-\I h m \tau_\ell} \,
  \Lambda(\e^{-\I h \tau_{\ell}})
  =
  0
\end{equation*}
for $m = 0, \dots, M-2L-1$.
Due to $\lambda_{2L+1} = 1$,
we may compute the remaining $\lambda_{\ell}$
by solving an equation system,
whose solution is, in fact, unique.
Knowing $\lambda_\ell$,
we determine the roots of $\Lambda$
and the frequencies $\tau_\ell$.
The coefficients $\gamma_\ell$ are then given
by the over-determined Vandermonde-type system
\begin{equation}
  \label{eq:vander-sys}
  \sum_{\ell=-L}^L \gamma_\ell \, \e^{-\I h m \tau_\ell}
  =
  E(hm),
  \quad
  m = 0, \dots, M.
\end{equation}
For our numerical experiments,
we use the so-called  Approximative Prony Method (APM)
by Potts and Tasche \cite{potts2010parameter},
which is based on the above consideration
but is numerically more stable.

\begin{algorithm}[APM, \cite{potts2010parameter}]
  \newline
  \emph{Input:} $L \in \BN$, $M \ge 4L+1$, $h>0$ with $h \tau_\ell < \pi$, $(E(hm))_{m=0}^M$.
  \begin{enumerate}[nosep]
  \item Compute the right singular vector $(\lambda_\ell)_{\ell=0}^{2L+1}$
    to the smallest singular value of
    \begin{equation*}
      (E(h(k+m)))_{m,k=0}^{M-2L-1,2L+1}.
    \end{equation*}
  \item Compute the roots $(z_\ell)_{\ell=-L}^L$ of
    $\Lambda(z)$ in \eqref{eq:prony-poly}
    in order $z_\ell = \bar z_{-\ell}$.
  \item Compute the least-square solution $(\gamma_\ell)_{\ell=-L}^L$ of \eqref{eq:vander-sys}.
  \item Set $\tau_\ell \coloneqq h^{-1} \arg z_\ell$.
  \end{enumerate}
  \emph{Output:} $\gamma_\ell$, $\tau_{\ell}$, $\ell = - L, \dots, L$.
\end{algorithm}

Instead of the exact number of terms $2L+1$,
the method can be applied with an upper estimate for $L$.
In this case,
$z_\ell$ not lying on the unit circle
and terms with small $\gamma_\ell$ can be neglected.

\section{The Phase Retrieval Problem}
\label{sec:phase-retr-probl}

Originally,
phase retrieval means the recovery of an unknown signal
only from the intensities of its Fourier transform.
In the following,
we consider phase retrieval for sparse signals,
which are superpositions of finitely many transitions of one known structure.
Moreover,
signal and structure are modeled as complex measures.
Thus,
the true signal $\mu \in \measure(\BR^D)$ is a superposition
of finitely many transitions $\nu_{T_n} \coloneqq \nu(\cdot - T_n)$
with $T_n \in \BR^D$ and $\nu \in \measure(\BR^D)$.
Using the convolution and the Dirac measure $\delta$,
the considered \emph{phase retrieval problem} has the following form:
recover the coefficients $c_n \in \BC \setminus \{0\}$
and the translations $T_n \in \BR^D$
of the \emph{structured signal}
\begin{equation}
  \label{eq:struc-sig}
  \mu
  \coloneqq
  \sum_{n=1}^{N} c_{n} \, \nu_{T_{n}}
  =
  \nu * \Bigl(
  \sum_{n=1}^{N} c_{n} \, \delta_{T_n}
  \Bigr)
\end{equation}
with $\nu \in \measure(\BR^D)$
from samples of its (squared) Fourier intensity
\begin{equation}
  \label{eq:four-int}
  \abs{\hat \mu(\omega)}^2
  =
  \abs{\hat \nu(\omega)}^2
  \sum_{n=1}^N \sum_{k=1}^N
  c_n \bar c_k \, \e^{-\I \ip{\omega}{T_n - T_k}}
  =
  \abs{\hat \nu(\omega)}^2
  \sum_{\ell = -L}^L
  \gamma_\ell \, \e^{-\I \ip{\omega}{\tau_\ell}}.
\end{equation}
The (indexed) families of translates and coefficients are henceforth denoted by
$\fT \coloneqq [T_1, \dots, T_N]$ and $\fC \coloneqq [c_1, \dots, c_N]$,
where the index of $c_n$ and $T_n$ always corresponds to each other.
Conceivable structures are the Dirac measure $\delta$,
in which case $\mu$ may be interpreted as Dirac signal,
or a Gaussian,
in which case $\mu$ may be interpreted as ordinary function
via its density function.
The considered phase retrieval problem is never uniquely solvable.

\begin{lemma}[Trivial Ambiguities]
  \label{lem:triv-amb}
  For every $\mu \in \measure(\BR^D)$,
  global phase shifts,
  transitions,
  and conjugated reflections
  have the same Fourier intensity,
  i.e.
  \begin{equation*}
    \abs{\hat \mu}
    =
    \abs{\fourier[ \e^{\I \alpha} \mu]}
    =
    \abs{\fourier[\mu(\cdot - x_0)]}
    =
    \abs{\fourier[\overline{\mu(- \cdot)}]},
    \quad
    \alpha \in \BR,
    \;
    x_0 \in \BR^D.
  \end{equation*}
\end{lemma}

Since the statement can be established
using standard computation rules of the Fourier transform,
the proof is omitted.
Besides these so-called trivial ambiguities,
there may occur further non-trivial ambiguities.

\begin{lemma}[Non-Trivial Ambiguities]
  \label{lem:non-triv-amb}
  Let $\mu = \mu_1 * \mu_2$ be the convolution of
  $\mu_1,\mu_2 \in \measure(\BR^D)$.
  The signal $\mu_1 * \overline{\mu_2(- \cdot)}$
  then has the same Fourier intensity.
\end{lemma}

The statement immediately follows form the Fourier convolution theorem.
As an immediate consequence,
we may conjugate and reflect
the structure part $\nu$
and the location part $\sum_{n=1}^N c_n \, \delta_{T_n}$
of \eqref{eq:struc-sig}
independently of each other.
Therefore,
we may only hope to recover $\mu$ up to the following ambiguities.

\begin{definition}[Inevitable Ambiguities]
  The signal in \eqref{eq:struc-sig}
  can only be recovered up to
  a global phase shift of all coefficients in $\fC$;
  a global shift of all transitions in $\fT$;
  and the conjugation of $\fC$ together with the reflection of $\fT$.
\end{definition}

In general,
the number of non-trivial ambiguities can be immense.
For example,
under the assumption that the transitions $\fT$ are equispaced,
there may exists up to $2^{N-2}$ non-trivial ambiguities \cite{beinert2015ambiguities},
which can be characterized as in \Cref{lem:non-triv-amb}.
To get rid of the non-trivial ambiguities,
we have to assume
that the transitions $\fT$ are somehow unregular. 

\begin{definition}[Collision-Freeness]
  A family $\fT \coloneqq [T_1, \dots, T_N] \subset \BR^D$
  is called \emph{collision-free}
  if the differences $T_n - T_k$ are pairwise distinct
  for all $n \ne k$.
\end{definition}

\section{Sparse Phase Retrieval on the Line}
\label{sec:phase-retrieval-line}

Phase retrieval of structured signals on the line
has been well studied \cite{beinert2017sparse,RCLV13,BP17,PPST18}.
Under certain assumptions,
the recovery from equispaced measurements is principally possible.

\begin{theorem}[Phase Retrieval, {\cite[Thm~3.1]{beinert2017sparse}}]
  \label{thm:pr1d}
  Let $\mu$ be of the form \eqref{eq:struc-sig}
  with $\hat \nu (\omega) \ne 0$, $\omega \in \BR$,
  collision-free $T_1 < \cdots < T_N$,
  and  $\abs{c_1} \neq \abs{c_N}$.
  Further,
  choose $h>0$ such that $h(T_n - T_k) < \pi$ for $n,k$.
  Then $\mu$ can be uniquely reconstructed
  from $\abs{\hat{\mu}(h m)}$, $m=0,\dots,2 N (N-1)+1$,
  up to inevitable ambiguities.
\end{theorem}

The constructive proof leads to a two-step algorithm:
First,
the parameters $\gamma_\ell$ and $\tau_\ell$ of \eqref{eq:four-int}
are determined using Prony's method.
Second,
the hidden relation between the indices $n,k$ and $\ell$
is revealed. 

\begin{algorithm}[Phase Retrieval on the Line, \cite{beinert2017sparse}]
  \label{alg:1d-pr}
  \newline
  \emph{Input:}
  $N\in\BN$,
  $M \ge 2 N (N-1) + 1$,
  $h > 0$ with $h (T_n - T_k) < \pi$,
  $(\abs{\hat \mu (hm)}^2)_{m=0}^M$.
  \begin{enumerate}[nosep]
  \item Apply APM with $E(hm) \coloneqq \abs{\hat\mu(hm)}^2 / \abs{\hat \nu(hm)}^2$
    to determine $(\tau_\ell)_{\ell=-L}^L$
    and $(\gamma_\ell)_{\ell=-L}^L$ in \eqref{eq:four-int}.
    Assume that $(\tau_\ell)_{\ell=-L}^L$ is ordered increasingly.
  \item Set $\fD \coloneqq  [ \tau_k : k=0,\dots, N(N-1)/2 ]$
    and $L \coloneqq N(N-1)/2$.
  \item Set $T_1 \coloneqq 0$,
    $T_N \coloneqq \tau_{L}$,
    $T_{N-1} \coloneqq \tau_{L-1}$. 
    Find $\ell^*$ with $|T_N-T_{N-1}| = \tau_{\ell^*}$.
    Set $c_1 \coloneqq \abs{\gamma_{L} \bar \gamma_{L-1}/\gamma_{\ell^*}}^{1/2}$,
    $c_N \coloneqq \gamma_{L} / \bar c_1$,
    $c_{N-1} \coloneqq \gamma_{L-1} / \bar c_1$.
    Remove $\tau_0$, $\tau_{\ell^*}$, $\tau_{L-1}$, $\tau_L$ from $\fD$.
  \item Initiate the lists $\fT \coloneqq [T_1, T_N, T_{N-1}]$
    and $\fC \coloneqq [c_1, c_N, c_{N-1}]$.
  \item For the maximal remaining $\tau_{L^*}$ in $\fD$,
    find $\ell^*$ with $\tau_{L^*}+ \tau_{\ell^*} = T_N$.
    Compute $d_{\mathrm r} \coloneqq \gamma_{L^*} / \bar c_1$
    and $d_{\mathrm l} \coloneqq \gamma_{\ell^*} / \bar c_1$.
    \begin{enumerate}[a)]
    \item If
      $
        \abs{c_N \bar d_{\mathrm r} - \gamma_{\ell^*}}
        <
        \abs{c_N \bar d_{\mathrm l} - \gamma_{L^*}}
      $,
      add $\tau_{L^*}$ to $\fT$
      and $d_{\mathrm r}$ to $\fC$.
    \item Otherwise add $\tau_{\ell^*}$ to $\fT$
      and $d_{\mathrm l}$ to $\fC$.
  \end{enumerate}
  Remove all $\abs{S - S'}$ with $S, S' \in \fT$ from $\fD$,
  and repeat until $\fD$ is empty.
\end{enumerate}
\emph{Output:} $\fC$, $\fT$.
\end{algorithm}

\section{Sparse Phase Retrieval on the Real Space}
\label{sec:phase-retrieval-real-space}

To extend the phase retrieval procedure
from the line to the real space,
we combine \Cref{thm:pr1d} and \Cref{alg:1d-pr}
with the sampling strategy in \cite{PW13}.
Instead of sampling the whole Fourier domain,
we only require the Fourier intensity $\abs{\hat \mu}$
sampled along adaptively chosen lines in $\BR^D$.
Initially,
we consider the setup
where we have given samples along the Cartesian axes
spanned by the unit vectors $e_d$, $d=1, \dots, D$,
and where we choose additional sampling lines accordingly.
Since the condition $\abs{c_1} \ne \abs{c_N}$ in \Cref{thm:pr1d}
strongly depends on the geometry of the transitions $T_n$,
we henceforth assume that
the absolute values $\abs{\fC}$ of the coefficients are distinct.

\begin{theorem}[Phase Retrieval on the Real Space]
  \label{thm:pr2d}
  Let $\mu$ be of the form \eqref{eq:struc-sig}
  with $\hat \nu(\omega) \ne 0$, $\omega \in \BR^D$,
  collision-free $\fT^{e_{d}}$ for every $d=1,\dots,D$,
  and distinct $\abs{\fC}$.
  Further,
  choose $h>0$ such that $h \, \norm{T_n - T_k} < \pi$ for all $n,k$.
  Then there exist $\theta_d\in \BR^d$, $d = 1, \dots, D-1$,
  with $\norm{\theta_d} = 1$
  such that $\mu$ can be uniquely reconstructed from
  \begin{equation*}
    \bigl\{
    \abs{\hat \mu(hm \, e_d)},
    :
    m = 0,\dots, 2N(N-1) + 1,
    d=1,\dots, D
    \bigr\}
  \end{equation*}
  and the adaptive samples
  \begin{equation*}
    \bigl\{
    \abs{\hat \mu(hm \, \theta_d)}
    :
    m = 0,\dots, 2N(N-1) + 1,
    d=1,\dots, D-1
    \bigr\}
  \end{equation*}
  up to inevitable ambiguities.
\end{theorem}

\begin{figure}[b]
  \centering
  \includegraphics{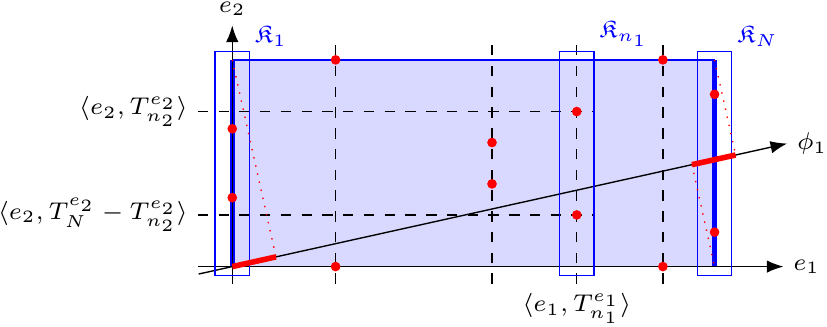}
  \caption{Construction of the candidate sets $\fK_{n_1}$ in the two-dimensional setup.
    Condition \eqref{eq:fix-order} is fulfilled
    if the projection of the left and right edge onto the line spanned by $\phi_1$%
    ---indicated by the red regions---%
    are disjoint.}
  \label{fig:cand}
\end{figure}

\begin{proof}
  Initially,
  we consider equispaced Fourier samples along an arbitrary line.
  For $\zeta \in \BR^D$ with $\norm{\zeta} = 1$,
  such samples have the form
  \begin{equation*}
    \abs{\hat \mu (hm \, \zeta)}^2
    =
    \abs{\hat \nu (hm \, \zeta)}^2
    \,
    \absB{\sum_{n=1}^{N} c_n \, \e^{-\I hm \ip{\zeta}{T_n}}}^2.
  \end{equation*}
  Essentially,
  these samples are the squared Fourier intensities
  of a point measure on the line
  with coefficient $c_{n}^{\zeta} \coloneqq c_n$
  and transitions $T_{n}^{\zeta} \coloneqq \ip{\zeta}{T_n}$.
  Without loss of generality,
  we henceforth assume that
  the families $\fC^{\zeta}\coloneqq [c_n^\zeta]$
  and $\fT^{\zeta}\coloneqq [T_n^\zeta]$ are ordered
  such that $0 = T_{1}^{\zeta} < \cdots < T_{N}^{\zeta}$.
  If $\fT^{\zeta}$ is collision-free,
  and $\fC^\zeta$ fulfils $\abs{c_1^\zeta} \ne \abs{c_N^\zeta}$,
  the assumption of \Cref{thm:pr1d} are satisfied,
  and $\fC^{\zeta}$ and $\fT^{\zeta}$ can be determined
  by \Cref{alg:1d-pr}
  up to trivial ambiguities.
  Where the global phase shifts and additional transitions are unproblematic,
  the conjugated reflection ambiguity has to be resolved.
  Owing to the distinct absolute values $\abs{\fC}$,
  we can identify the transitions $T_n^{e_d}$ with each other.
  More precisely,
  for every $n_1 \in \{1, \dots, N\}$,
  we find unique indices $n_2, \dots, n_D$
  such that $\abs{c_{n_1}^{e_1}} = \cdots = \abs{c_{n_D}^{e_D}}$.
  Up to a global shift,
  the true transitions $T_n$ of \eqref{eq:struc-sig}
  are contained in
  \begin{equation}
    \label{eq:cand-set}
    \tilde \fT
    \coloneqq
    \fK_1 \cup \cdots \cup \fK_N.
  \end{equation}
  with the candidate sets
  \begin{equation*}
    \fK_{n_1}
    \coloneqq
    \{
    (T_{n_1}^{e_1}, R_{n_2}^{e_2}, \dots, R_{n_D}^{e_D})
    :
    R_{n_d}^{e_d} \in \{ T^{e_{d}}_{n_d}, (T^{e_{d}}_{N} - T^{e_{d}}_{n_d}) \}
    \}
  \end{equation*}
  where $\abs{c_{n_1}^{e_1}} = \cdots = \abs{c_{n_D}^{e_D}}$.
  The construction of $\tilde \fT$ is schematically illustrated in Fig.~\ref{fig:cand}.
  Notice that the candidate sets $\fK_{n_1}$ are the vertices
  of $(D-1)$-dimensional cubes lying
  on shifted versions of the hyperplane orthogonal to $e_1$. 
  Therefore,
  we find $\phi_d$ with $\norm{\phi_d} = 1$
  such that $\{e_1, \phi_1, \dots, \phi_{D-1}\}$ form a basis
  and that
  \begin{equation}
    \label{eq:fix-order}
    \ip{\phi_d}{S_1} \le \ip{\phi_d}{S_N}
    \quad\text{for all}\quad
    S_1 \in \fK_1, S_N \in \fK_N,
  \end{equation}
  see also Fig.~\ref{fig:cand}.
  Moreover,
  $\phi_d$ may be chosen
  so that $\fT^{\phi_d}$ is collision-free,
  and Algorithm~\ref{alg:1d-pr} can be applied to recover $\fC^{\phi_d}$ and $\fT^{\phi_d}$.
  Due to \eqref{eq:fix-order},
  the conjugated reflection ambiguity for $\fT^{\phi_d}$ can be resolved,
  and the true transitions are given by
  \begin{equation}
    \label{eq:rec-T}
    \Bigl\{
    \bigl((e_1 \vert \phi_1 \vert \dots \vert \phi_D)^{*}\bigr)^{-1}
    (T_{n_1}^{e_1}, T_{n_2}^{\phi_1}, \dots, T_{n_D}^{\phi_{D-1}})^*
    :
    n_1 = 1, \dots, N
    \Bigr\},
  \end{equation}
  where $*$ denotes the conjugation and transposition,
  and where the indices are uniquely given by
  $\abs{c_{n_1}^{e_1}} = \abs{c_{n_2}^{\phi_1}} = \cdots = \abs{c_{n_D}^{\phi_{D-1}}}$.
  The reflection of $\fT^{e_1}$ would yield the conjugated reflection of $\mu$,
  which concludes the proof.
  \qed
\end{proof}

Similarly to the recovery guarantee on the line (\Cref{thm:pr1d}),
the proof is constructive and summerizes to the following reconstruction method,
which we will also use for the numerical simulations in §~\ref{sec:simulations}.

\begin{algorithm}[Phase Retrieval on the Real Space]
  \label{alg:md-pr}
  \newline
  \emph{Input:}
  $D \in \BN$,
  $N \in \BN$,
  $M \ge 2N(N-1) + 1$,
  $h > 0$ with $h \norm{T_{n}- T_{k}} < \pi$,
  adaptive samples of $\abs{\hat \mu}^{2}$.
  \begin{enumerate}[nosep]
  \item Sample $\abs{\hat\mu}^{2}$ along the axes to obtain
    $(\abs{\hat\mu(hm \, e_{d})}^2)_{m=0}^{M}$
    for $d=1,\dots,D$.
  \item Apply \Cref{alg:1d-pr} to compute $\fT^{e_{d}}$,
    and build $\tilde \fT$ in \eqref{eq:cand-set}.
  \item For $d=1,\dots,D-1$,
    choose $\theta_d \in \BR^D$ randomly
    such that \eqref{eq:fix-order} is satisfied.
  \item Sample $\abs{\hat\mu}^2$ to obtain
    $(\abs{\hat\mu(hm \, \theta_d)}^{2})_{m=0}^{M}$.
  \item Apply \Cref{alg:1d-pr} to compute
    $\fC^{\theta_d}$ and $\fT^{\theta_d}$.
  \item Compute $\fT$ by solving the equation system in \eqref{eq:rec-T},
    and set $\fC \coloneqq \fC^{e_1}$.
  \end{enumerate}
  \emph{Output:} $\fC$, $\fT$.
\end{algorithm}

The statement of \Cref{thm:pr2d} remains valid
if the unit vectors $e_d$ are replaced by an arbitrary basis $\psi_d$.
Similarly to \eqref{eq:rec-T},
the candidate set $\tilde \fT$ can then be determined by 
\begin{equation}
  \label{eq:gen-cand-set}
  \tilde \fT
  \coloneqq
  \bigl((\psi_1 \vert \dots \vert \psi_D)^{*}\bigr)^{-1}
  (\fK_1 \cup \cdots \fK_N).
\end{equation}
One of the key assumptions in \Cref{thm:pr2d} is that
the coordinates of $\fT$ with respect to the considered directions
are collision-free.
For a given set of transitions $\fT$,
this is holds true for almost all directions $\phi \in \BR^D$,
i.e.\ up to a Lebesgue null set.

\begin{lemma}
  \label{lem:coll-free}
  Let $\fT \coloneqq [T_1, \dots, T_N] \subset \BR^D$ be collision-free,
  then the family $\fT^\theta \coloneqq [\ip{\theta}{T_n} : T_n \in \fT]$
  is collision-free for almost all $\theta \in \BR^D$. 
\end{lemma}

\begin{proof}
  Let $T_{i_1}$, $T_{i_2}$, $T_{i_3}$, $T_{i_4}$ be
  arbitrary points in $\fT$,
  where only $T_{i_2}$ and $T_{i_3}$ may coincide.
  The dimension of the subspace
  $\{\theta \in \BR^D : \ip{\theta}{T_{i_1}-T_{i_2}} = \ip{\theta}{T_{i_3}-T_{i_4}}\}$
  may be $D-1$ at the most
  since $\fT$ is collision-free.
  Thus the family $\fT^\theta$ can only contain collisions
  for $\theta$ lying in the union of finitely many lower-dimensional subspaces,
  which gives the assertion.
  \qed
\end{proof}

Against the background of \Cref{lem:coll-free},
and since $D$ generic vectors form a basis,
the sampling along adaptively chosen lines
in the two-dimensional setup
can be replaced
by the sampling along arbitrary generic lines.

\begin{theorem}[Phase Retrieval in 2D]
  Let $\mu$ be of the form \eqref{eq:struc-sig}
  with $\hat \nu(\omega) \ne 0$,
  $\omega \in \BR^2$,
  collision-free $\fT$,
  and distinct $\abs{\fC}$.
  Further,
  let $\psi_1, \psi_2, \psi_{3}$ be generic vectors in $\BR^2$.
  Then $\mu$ can be uniquely reconstructed from
  \begin{equation*}
    \bigl\{
    \abs{\hat \mu(hm \, \psi_d)}
    :
    m = 0,\dots, 2N(N-1) + 1,
    d=1,\dots, 3
    \bigr\}
  \end{equation*}
  up to inevitable ambiguities.
\end{theorem}

\begin{proof}
  To establish the statement,
  we adapt the proof of Theorem~\ref{thm:pr2d},
  where $\psi_1$, $\psi_2$ play the role of $e_1$, $e_2$,
  and $\psi_3$ the role of the adaptive direction $\phi_1$.
  Due to Lemma~\ref{lem:coll-free},
  the families $\fT^{\phi_d}$ are collision-free for generic lines;
  so the application of Algorithm~\ref{alg:1d-pr} is unproblematic.
  Considering the construction of $\fK_{n_1}$,
  we notice that the candidates are contained in a parallelogram,
  where $\fK_1$ and $\fK_N$ lie on opposite edges,
  see Fig.~\ref{fig:cand} for an illustration.
  If the coordinates of $\fK_1$ and $\fK_N$ satisfy \eqref{eq:fix-order}
  with respect to the third direction $\psi_3$,
  we can apply the procedure in the proof of Theorem~\ref{thm:pr2d}
  to recover $\mu$.
  Otherwise,
  we interchange the role of $\psi_1$ and $\psi_2$.
  In this case,
  we obtain a second set of candidates in the same parallelogram.
  The new sets $\fK_1'$ and $\fK_N'$, however, lie on the remaining to two edges
  such that $\psi_3$ now fulfils \eqref{eq:fix-order},
  and $\mu$ can be recovered.
  \qed
\end{proof}

\section{Simulations}
\label{sec:simulations}

To substantiate the theoretical observations,
we apply the constructive proof (Algorithm~\ref{alg:md-pr})
to a minor, synthetic example.
Our goal is recover the sparse structure in Fig.~\ref{fig:sim-mu},
which consists of five sources.
Each source corresponds to a Gaussian
with standard derivation $\nicefrac{1}{2}$
and to a complex coefficient.
For repeatability,
the exact locations and coefficients are given in Table~\ref{tab:T-C}.
The Fourier intensity of the true signal is shown in Fig.~\ref{fig:sim-mu-hat}.
Instead of sampling the whole Fourier domain,
we only use equispaced samples on three predefined lines,
which are depicted as red lines in Fig.~\ref{fig:sim-mu-hat}.
The first two lines correspond to the Cartesian axes,
and the third to the angle $0.143 \pi$.
On each line,
we take 100 equispaced, noise-free samples,
which is slightly more than the 83 samples to employ
the univariate, sparse phase retrieval in Algorithm~\ref{alg:1d-pr}.
The sampling distance is
$h \coloneqq \max\{ \norm{T_k - T_n}\} / 2\pi \approx 0.0387$.

\begin{figure}
  \centering
  \subcaptionbox{Density of $\abs{\mu}$.\label{fig:sim-mu}}{%
    \includegraphics{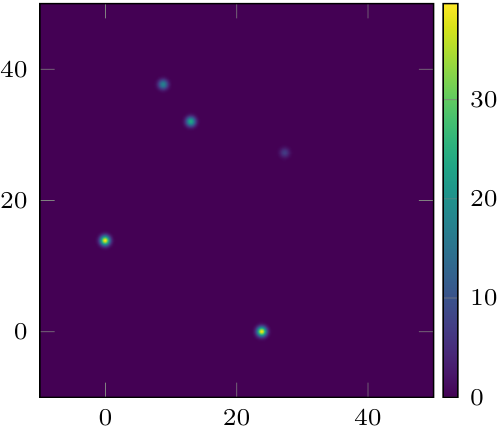}}
  \quad
  \subcaptionbox{Fourier intensity $\abs{\hat \mu}$\label{fig:sim-mu-hat}}{%
    \includegraphics{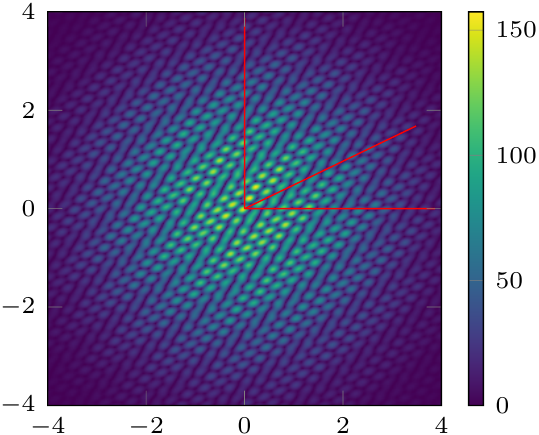}}
  \caption{Absolute value of the complex-valued density of the true $\mu$
    and the corresponding Fourier intensity $\abs{\hat \mu}$.
    The red lines indicate the predefined sampling direction for Algorithm~\ref{alg:md-pr}.}
  \label{fig:sim}
\end{figure}

\begin{table}[t]
  \centering\footnotesize
  \caption{Transitions and coefficients of the true sparse signal $\mu$.
    Algorithm~\ref{alg:md-pr} recovers $\fT$ up to a maximal absolute error of $6.982 \cdot 10^{-8}$,
    and $\fC$ up to a maximal absolute error of $2.898 \cdot 10^{-5}$.}
  \label{tab:T-C}
  \begin{tabular}{lccccc}
    \toprule
    $n$ & 1 & 2 & 3 & 4 & 5 \\
    \midrule
    $T_n^*$ & $(27.374, 27.258)$ & $(13.065, 32.008)$ & $(8.847, 37.665)$ & $(0.000, 13.874)$ & $(23.876, 0.000)$ \\
    $c_n$ & $7.293+5.115\I$ & $30.665 +2.258\I$ & $2.740 +22.286\I$ & $1.576+49.834\I$ & $17.400  +46.587\I$ \\
    \bottomrule
  \end{tabular}
\end{table}

\begin{figure}[t]
  \centering
  \includegraphics{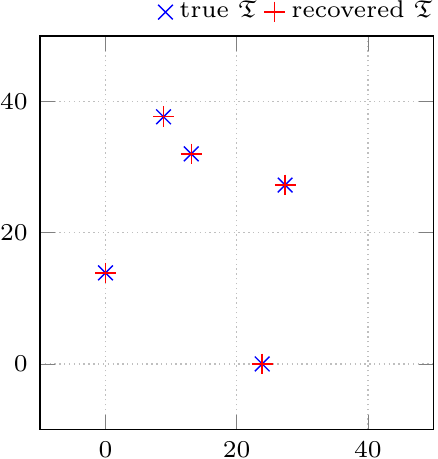}
  \caption{True and recovered transitions $\fT$ and $\fT'$
    after removing the trivial shift and conjugated reflection ambiguity.}
  \label{fig:rec-T}
\end{figure}

Applying Algorithm~\ref{alg:md-pr} here yields
an accurate approximation of the true transitions $\fT$.
To compare the recovered values $\fT'$ with the true ones,
we conjugate and reflect the recovered signal
such that both signals have the same orientation.
Furthermore,
we shift both signals
so that the leftmost and lowermost vectors
lie on the Cartesian axes.
The recovered values are shown in Fig.~\ref{fig:rec-T}
and nearly coincide with the true values.
The maximal absolute error
after eliminating the shift and conjugated reflection ambiguity 
is given by
\begin{equation*}
  \max_{n=1,\dots,5}\{\abs{T_n - T'_n} \}
  =
  6.982 \cdot 10^{-8}.
\end{equation*}
The maximal absolute error for the coefficients is here
\begin{equation*}
  \max_{n=1,\dots,5}\{\abs{c_n - \e^{\I \alpha} \, c'_n} \}
  =
  2.898 \cdot 10^{-5},
\end{equation*}
where we choose the global phase shift $\alpha$
such that the maximal absolute error is minimized.

This first numerical example shows that
a sparse, bivariate signal can completely recovered
using only samples on three predefined lines
instead of sampling the whole Fourier domain.
In all fairness,
Prony's method is known to be very sensitive to noise
and to become unstable
if the frequencies---%
here $\ip{\zeta}{T_n - T_k}$---%
almost coincide.
Since Prony's method has to recover
$2N(N-1)+1$ frequencies for an $N$-sparse signal,
$N$ cannot be increased very far.
The development of a stable algorithm
for the proposed sampling scheme,
which can deal with additional noise,
thus remains open for  further research.

\section{Conclusion}
\label{sec:conclusion}

The focus of this paper is the sparse phase retrieval problem
as it occurs in speckle imaging and crystallography.
Combining the one-dimensional phase retrieval approach in \cite{PPST18,beinert2017sparse}
with the adaptive sampling strategy from \cite{PW13},
we derive a Prony-based recovery method for $D$-variate, sparse signals.
This constructive method is the key ingredient
to establish recovery guarantees for collision-free signals.
Instead of sampling the Fourier intensity on the entire $D$-dimensional domain,
we only require samples from $2D-1$ adaptively chosen or generic lines.
This leads to a sampling complexity of $\mathcal O(DN^2)$,
where $D$ is the dimension of the signal
and $N$ the sparsity level,
i.e.\ the number of shifted components.
The $N^2$ is accounted for by Prony's method,
which we use to determine the parameters of the given Fourier intensity%
---an exponential sum with quadratic structure.
Since the unknown signal consists of merely $(D+2)N$ real parameters,
the question arises
whether the signal can also be uniquely recovered
using less measurements.
This question is left for further research.
During the numerical simulations,
we show that the Prony-based method can, in principle, be applied
to recover the wanted signal from noise-free measurements.
Since Prony's method is, however, very sensitive to noise,
one of our next steps is to derive a more stable algorithm
for the recovery of specific structured, multivariate signals.


\bibliographystyle{splncs04}
\bibliography{reference}

\end{document}